\documentclass[11pt,runningheads,a4paper]{llncs}

\usepackage[T1]{fontenc}
\usepackage{graphicx}
\usepackage{amsmath}
\usepackage{hyperref}
\usepackage{doi}

\def\Oh{\mathcal{O}}

\def\X{\mathsf{X}}
\def\LP{\mbox{\rm {\sf L}}}
\def\Q{\mbox{\rm {\sf Q}}}
\def\R{\mbox{\rm {\sf R}}}
\def\Y{\mathsf{Y}}
\def\LPF{\mbox{\rm {\sf LPF}}}
\def\LZ{\mbox{\rm {\sf LZ}}}

\def\abstractname{\normalsize Abstract.}

\begin{document}
\title{Lempel-Ziv Decoding in External Memory\thanks{This research is
    partially supported by Academy of Finland through grant 258308 and
    grant 250345 (CoECGR).}}

\author{
  Djamal Belazzougui\inst{1}
  \and
  Juha K{\"a}rkk{\"a}inen\inst{2}
  \and
  Dominik Kempa\inst{2}
  \and\\
  Simon J. Puglisi\inst{2}
}

\authorrunning{D. Belazzougui et al.}

\institute{
  CERIST, Algeria\\
  \email{dbelazzougui@cerist.dz}\\[0.1cm]
  \and
  Helsinki Institute for Information Technology HIIT, \\
  Department of Computer Science, University of Helsinki, Finland\\
  \email{\{juha.karkkainen,dominik.kempa,simon.puglisi\}@cs.helsinki.fi}\\
  \vspace{-0.2cm}
}

\date{}

\maketitle \thispagestyle{empty}

\begin{abstract}
  \normalsize Simple and fast decoding is one of the main advantages
  of LZ77-type text encoding used in many popular file compressors
  such as {\em gzip} and {\em 7zip}. With the recent introduction of
  external memory algorithms for Lempel--Ziv factorization there is a
  need for external memory LZ77 decoding but the standard algorithm
  makes random accesses to the text and cannot be trivially modified
  for external memory computation. We describe the first external
  memory algorithms for LZ77 decoding, prove that their I/O complexity
  is optimal, and demonstrate that they are very fast in practice,
  only about three times slower than in-memory decoding (when reading
  input and writing output is included in the time).
\end{abstract}

\section{Introduction}
\label{sec-intro}

The Lempel--Ziv (LZ) factorization~\cite{LZ76} is a partitioning of a
text string into a minimal number of phrases consisting of substrings
with an earlier occurrence in the string and of single characters. In
LZ77 encoding~\cite{ZL77} the repeated phrases are replaced by a
pointer to an earlier occurrence (called the source of the phrase). It
is a fundamental tool for data
compression~\cite{fm2010,RLZspire2010,GG10,kn2010} and today it lies
at the heart of popular file compressors (e.g. {\em gzip} and {\em
  7zip}), and information retrieval systems~(see,
e.g.,~\cite{fm2010,hpz2011}). Recently the factorization has become
the basis for several compressed full-text
self-indexes~\cite{kn2011,ggp2011,ggknp2012,fghp2013}.  Outside of
compression, LZ factorization is a widely used algorithmic tool in
string processing: the factorization lays bare the repetitive
structure of a string, and this can be used to design efficient
algorithms~\cite{bct2012,kk1999,kk2003,kbk2003}.

One of the main advantages of LZ77 encoding as a compression technique
is a fast and simple decoding: simply replace each pointer to a source
by a copy of the source. However, this requires a random access to the
earlier part of the text. Thus the recent introduction of external
memory algorithms for LZ77 factorization~\cite{KKP2014-EMLZ} raises
the question: Is fast LZ77 decoding possible when the text length
exceeds the RAM size? In this paper we answer the question positively
by describing the first external memory algorithms for LZ77 decoding.

In LZ77 compression, the need for external memory algorithms can be
avoided by using an encoding window of limited size. However, a longer
encoding window can improve the compression ratio~\cite{fm2010}. Even
with a limited window size, decompression on a machine with a small
RAM may require an external memory algorithm if the compression was
done on a machine with a large RAM. Furthermore, in applications such
as text indexing and string processing limiting the window size is not
allowed. While most of these applications do not require decoding, a
fast decoding algorithm is still useful for checking the correctness
of the factorization.

\paragraph{Our contribution.} 
We show that in the standard external memory model~\cite{Vit2006} the
I/O complexity of decoding an LZ77-like encoding of a string of length
$n$ over an alphabet of size $\sigma$ is $\Theta\left( \frac{n}{B
    \log_\sigma n} \log_{M/B} \frac{n}{B \log_\sigma n} \right)$,
where $M$ is the RAM size and $B$ is the disk block size in units of
$\Theta(\log n)$ bits.  The lower bound is shown by a reduction from
permuting and the upper bound by describing two algorithms with this
I/O complexity.

The first algorithm uses the powerful tools of external memory sorting
and priority queues while the second one relies on plain disk I/O
only. Both algorithms are relatively simple and easy to implement.
Our implementation uses the STXXL library~\cite{STXXL} for sorting and
priority queues.

Our experiments show that both algorithms scale well for large data
but the second algorithm is much faster in all cases. This shows that,
while external memory sorting and priority queues are extremely useful
tools, they do have a significant overhead when their full power is
not needed.  The faster algorithm (using a very modest amount of RAM)
is only 3--4 times slower than an in-memory algorithm that has enough
RAM to perform the decoding in RAM (but has to read the input from
disk and write the output to disk).

Our algorithms do not need a huge amount of disk space in addition to
the input (factorization) and output (text), but we also describe and
implement a version, which can reduce the additional disk space to
less than 3\% of total disk space usage essentially with no effect on
runtime.

\section{Basic Definitions}
\label{sec-preliminaries}

\paragraph{Strings.}
Throughout we consider a string $\X = \X[1..n] = \X[1]\X[2]\ldots
\X[n]$ of $|\X| = n$ symbols drawn from the alphabet $[0..\sigma-1]$
for $\sigma = n^{\Oh(1)}$.  For $1 \leq i \leq j \leq n$ we write
$\X[i..j]$ to denote the {\em substring} $\X[i]\X[i+1]\ldots \X[j]$ of
$\X$.  By $\X[i..j)$ we denote $\X[i..j-1]$.

\paragraph{LZ77.}
The {\em longest previous factor} (LPF) at position $i$ in string $\X$
is a pair $\LPF[i]=(p_i,\ell_i)$ such that, $p_i < i$,
$\X[p_i..p_i+\ell_i) = \X[i..i+\ell_i)$, and $\ell_i$ is maximized.
In other words, $\X[i..i+\ell_i)$ is the longest prefix of $\X[i..n]$
which also occurs at some position $p_i < i$ in $\X$.  There may be
more than one potential value of $p_i$, and we do not care which one
is used.

The LZ77 factorization (or LZ77 parsing) of a string $\X$ is a greedy,
left-to-right parsing of $\X$ into longest previous factors. More
precisely, if the $j$th LZ factor (or {\em phrase}) in the parsing is
to start at position $i$, then $\LZ[j] = \LPF[i] = (p_i,\ell_i)$ (to
represent the $j$th phrase), and then the $(j+1)$th phrase starts at
position $i+\ell_i$. The exception is the case $\ell_i=0$, which
happens iff $\X[i]$ is the leftmost occurrence of a symbol in $\X$. In
this case $\LZ[j] = (\X[i],0)$ (to represent $\X[i..i]$) and the next
phrase starts at position $i+1$.  This is called a \emph{literal
  phrase} and the other phrases are called \emph{repeat phrases}.  For
a repeat phrases, the substring $\X[p_i..p_i+\ell_i)$ is called the
{\em source} of the phrase $\X[i..i+\ell_i)$.  We denote the number of
phrases in the LZ77 parsing of $\X$ by $z$.

\paragraph{LZ77-type factorization.}
There are many variations of LZ77 parsing. For example, the original
LZ77 encoding~\cite{ZL77} had only one type of phrase, a (potentially
empty) repeat phrase always followed by a literal character. Many
compressors use parsing strategies that differ from the greedy
strategy described above to optimize compression ratio after entropy
compression or to speed up compression or decompression.  The
algorithms described in this paper can be easily adapted for most of
them.  For purposes of presentation and analysis we make two
assumptions about the parsing:
\begin{itemize}
\item All phrases are either literal or repeat phrases as described
  above.
\item The total number of repeat phrases, denoted by
  $z_{\mathrm{rep}}$, is $\Oh(n/\log_\sigma n)$.
\end{itemize}
We call this an \emph{LZ77-type factorization}. The second assumption
holds for the greedy factorization~\cite{LZ76} and means that the
total size of the repeat phrases encoded using $\Oh(\log n)$-bit
integers is $\Oh(n\log\sigma)$. If furthermore the zero length in the
literal phrases is encoded with $\Oh(\log\sigma)$ bits, the size of
the whole encoding is $\Oh(n\log\sigma)$ bits.

\section{On I/O complexity of LZ decoding}

Given an LZ77-type factorization of a string encoded as described
above, the task of LZ77 decoding is to recover the original string. In
this section, we obtain a lower bound on the I/O complexity of LZ
decoding by a reduction from permuting.

We do the analysis using the standard external memory (EM)
model~\cite{Vit2006} with RAM size $M$ and disk block size $B$, both
measured in units of $\Theta(\log n)$ bits. We are primarily
interested in the I/O complexity, i.e., the number of disk blocks
moved between RAM and disk.

Given a sequence $\bar{x}=x_1, x_2, \dots, x_n$ of $n$ objects of size
$\Theta(\log n)$ bits each and a permutation $\pi[1..n]$ of $[1..n]$,
the task of permuting is to obtain the permuted sequence $\bar{y} =
y_1, y_2, \dots, y_n = x_{\pi[1]}, x_{\pi[2]}, \dots,
x_{\pi[n]}$. Under the mild assumption that $B\log (M/B) = \Omega(\log
(n/B))$, the I/O complexity of permuting is $\Theta\left( \frac{n}{B}
  \log_{M/B}\frac{n}{B} \right)$, the same as the I/O complexity of
sorting~\cite{AV88}.

We show now that permuting can be reduced to LZ decoding.  Let $\X$ be
the string obtained from the sequence $\bar{x}$ by encoding each $x_i$
as a string of length $h=\Theta(\log_\sigma n)$ over the alphabet
$[0..\sigma)$.  Let $\Y$ be the string obtained in the same way from
the sequence $\bar{y}$.  Form an LZ77-type factorization of $\X\Y$ by
encoding the first half using literal phrases and the second half
using repeat phrases so that the substring representing $y_i$ is
encoded by the phrase $(h\pi[i]+1-h,h)$. This LZ factorization is easy
to construct in $\Oh(n/B)$ I/Os given $\bar{x}$ and $\pi$. By decoding
the factorization we obtain $\X\Y$ and thus $\bar{y}$.

\begin{theorem}
  The I/O complexity of decoding an LZ77-type factorization of a
  string of length $n$ over an alphabet of size $\sigma$ is
  $\Omega\left(\frac{n}{B\log_\sigma n} \log_{M/B}
    \frac{n}{B\log_\sigma n}\right)$.
\end{theorem}

\begin{proof}
  The result follows by the above reduction from permuting a sequence
  of $\Theta(n/\log_\sigma n)$ objects.  \qed
\end{proof}

\section{LZ decoding using EM sorting and priority queue}
\label{sec:decode-using-sort-and-pq}

Our first algorithm for LZ decoding relies on the powerful tools of
external memory sorting and external memory priority queues.

We divide the string $\X$ into $\lceil n/b\rceil$ segments of size
exactly $b$ (except the last segment can be smaller). The segments
must be small enough to fit in RAM and big enough to fill at least one
disk block. If a phrase or its source overlaps a segment boundary, the
phrase is split so that all phrases and their sources are completely
inside one segment.  The number of phrases increases by at most
$\Oh(z_{\mathrm{rep}}+n/b)$ because of the splitting.

After splitting, the phrases are divided into three sequences.  The
sequence $\R_{\mathrm{far}}$ contains repeat phrases with the source
more than $b$ positions before the phrase (called far repeat phrases)
and the sequence $\R_{\mathrm{near}}$ the other repeat phrases (called
near repeat phrases). The sequence $\LP$ contains all the literal
phrases. The repeat phrases are represented by triples $(p,q,\ell)$,
where $p$ is the starting position of the source, $q$ is the starting
position of the phrase and $\ell$ is the length.  The literal phrases
are represented by pairs $(q,c)$, where $q$ is the phrase position and
$c$ is the character.  The sequence $\R_{\mathrm{far}}$ of far repeat
phrases is sorted by the source position.  The other two sequences are
not sorted, i.e., they remain ordered by the phrase position.

During the computation, we maintain an external memory priority queue
$\Q$ that stores already recovered far repeat phrases. Each such
phrase is represented by a triple $(q,\ell,s)$, where $q$ and $\ell$
are as above and $s$ is the phrase as a literal string. The triples
are extracted from the queue in the ascending order of $q$. The
maximum length of phrases stored in the queue is bounded by a
parameter $\ell_{\max}$. Longer phrases are split into multiple
phrases before inserting them into the queue.

The string $\X$ is recovered one segment at a time in left-to-right
order and each segment is recovered one phrase at a time in
left-to-right order. A segment recovery is done in a (RAM) array
$\Y[0..b)$ of size $b$. At any moment in time, for some $i\in[0..b]$,
$\Y[0..i)$ contains the already recovered prefix of the current
segment and $\Y[i..b)$ contains the last $b-i$ characters of the
preceding segment. The next phrase starting at $\Y[i]$ is recovered in
one of three ways depending on its type:
\begin{itemize}
\item A literal phrase is obtained as the next phrase in the sequence
  $\LP$.
\item A near repeat phrase is obtained as the next phrase in the
  sequence $\R_{\mathrm{near}}$. The source of the phrase either
  starts in $\Y[0..i)$ or is contained in $\Y[i..b)$, and is easily
  recovered in both cases.
\item A far repeat phrase is obtained from the the priority queue with
  the full literal representation.
\end{itemize}

Once a segment has been fully recovered, we read all the phrases in
the sequence $\R_{\mathrm{far}}$ having the source within the current
segment.  Since $\R_{\mathrm{far}}$ is ordered by the source position,
this involves a single sequential scan of $\R_{\mathrm{far}}$ over the
whole algorithm. Each such phrase is inserted into the priority queue
$\Q$ with its literal representation (splitting the phrase into
multiple phrases if necessary).

\begin{theorem}
  A string of length $n$ over an alphabet of size $\sigma$ can be
  recovered from its LZ77 factorization in $\Oh\left(
    \frac{n}{B\log_\sigma n} \log_{M/B}\frac{n}{B\log_\sigma n}
  \right)$ I/Os.
\end{theorem}

\begin{proof}

  We set $\ell_{\max}=\Theta(\log_\sigma n)$ and $b=\Theta(B
  \log_\sigma n)$. Then the objects stored in the priority queue need
  $\Oh(\log n + \ell_{\max} \log\sigma) = \Oh(\log n)$\linebreak bits
  each and the total number of repeat phrases after all splitting is
  $\Oh(z_{\mathrm{rep}} + n/\log_\sigma n) = \Oh(n/\log_\sigma n)$.
  Thus sorting the phrases needs\linebreak $\Oh\left(
    \frac{n}{B\log_\sigma n} \log_{M/B}\frac{n}{B\log_\sigma n}
  \right)$ I/Os. This is also the I/O complexity of all the external
  memory priority queue operations~\cite{BK98}.  All other processing
  is sequential and needs $\Oh\left( \frac{n}{B\log_\sigma n}\right)$
  I/Os.  \qed

\end{proof}

We have implemented the algorithm using the STXXL library~\cite{STXXL}
for external memory sorting and priority queues.

\section{LZ decoding without sorting or priority queue}
\label{sec:decode-using-plain-disk-io}

The practical performance of the algorithm in the previous section is
often bounded by in-memory computation rather than I/O, at least on a
machine with relatively fast disks. In this section, we describe an
algorithm that reduces computation based on the observation that we do
not really need the full power of external memory sorting and priority
queues.

To get rid of sorting, we replace the sorted sequence
$\R_{\mathrm{far}}$ with $\lceil n/b \rceil$ unsorted sequences $\R_1,
\R_2, \dots$, where $\R_i$ contains all phrases with the source in the
$i$th segment. In other words, sorting $\R_{\mathrm{far}}$ is replaced
with distributing the phrases into the segments $\R_1, \R_2,
\dots$. If $n/b$ is less than $M/B$, the distribution can be done in
one pass, since we only need one RAM buffer of size $B$ for each
segment. Otherwise, we group $M/B$ consecutive segments into a
supersegment, distribute the phrases first into supersegments, and
then scanning the supersegment sequences into segments. If necessary,
further layers can be added to the segment hierarchy. This operation
generates the same amount of I/O as sorting the phrases but requires
less computation because the segment sequences do not need to be
sorted.

In the same way, the priority queue is replaced with $\lceil n/b
\rceil$ simple queues.  The queue $\Q_i$ contains a triple
$(q,\ell,s)$ for each far repeat phrase whose phrase position is
within the $i$th segment.  The order of the phrases in the queue is
arbitrary.  Instead of inserting a recovered far repeat phrase into
the priority queue $\Q$ it is appended into the appropriate queue
$\Q_i$. This requires a RAM buffer of size $B$ for each queue but as
above a multi-round distribution can be used if the number of segments
is too large. This approach may not reduce the I/O compared to the use
of a priority queue but it does reduce computation. Moreover, the
simple queue allows the strings $s$ to be of variable sizes and of
unlimited length; thus there is no need to split the phrases except at
segment boundaries.

Since the queues $\Q_i$ are not ordered by the phrase position, we can
no more recover a segment in a strict left-to-right order, which
requires a modification of the segment recovery procedure.  The
sequence $\R_{\mathrm{near}}$ of near repeat phrases is divided into
two: $\R_{\mathrm{prev}}$ contains the phrases with the source in the
preceding segment and $\R_{\mathrm{same}}$ the ones with the source in
the same segment.

As in the previous section, the recovery of a segment $\X_j$ starts
with the previous segment in the array $\Y[0..b)$ and consists of the
following steps:
\begin{enumerate}
\item Recover the phrases in $\R_{\mathrm{prev}}$ (that are in this
  segment). Note that each source is in the part of the previous
  segment that is still untouched.
\item Recover the literal phrases by reading them from $\LP$.
\item Recover the far repeat phrases by reading them from $\Q_j$ (with
  the full literal representation).
\item Recover the phrases in $\R_{\mathrm{same}}$. Note that each
  source is in the part of the current segment that has been fully
  recovered.
\end{enumerate}
After the recovery of the segment, we read all the phrases in $\R_j$
and insert them into the queues $\Q_k$ with their full literal
representations.

We want to minimize the number of segments. Thus we choose the
segments size to occupy at least half of the available RAM and more if
the RAM buffers for the queues $\Q_k$ do not require all of the other
half.  It is easy to see that this algorithm does not generate
asymptotically more I/Os than the algorithm of the previous
section. Thus the I/O complexity is $\Oh\left( \frac{n}{B\log_\sigma
    n} \log_{M/B}\frac{n}{B\log_\sigma n} \right)$.  We have
implemented the algorithm using standard file I/O (without the help of
STXXL).

\section{Reducing disk space usage}
\label{sec:reducing-disk-space-usage}

The algorithm described in the previous section can adapt to a small
RAM by using short segments, and if necessary, multiple rounds of
distribution.  However, reducing the segment size does not affect the
disk space usage and the algorithm will fail if it does not have
enough disk space to store all the external memory data. In this
section, we describe how the disk space usage can be reduced.

The idea is to divide the LZ factorization into parts and to process
one part at a time recovering the corresponding part of the text.  The
first part is processed with the algorithm of the previous section as
if it was the full string. To process the later parts, a slightly
modified algorithm is needed because, although all the phrases are in
the current part, the sources can be in the earlier parts. Thus we
will have the $\R_j$ queues for all the segments in the current and
earlier parts but the $\Q_j$ queues only for the current part.  The
algorithm processes first all segments in the previous parts
performing the following steps for each segment $\X_j$:
\begin{itemize}
\item Read $\X_j$ from disk to RAM.
\item Read $\R_j$ and for each phrase in $\R_j$ create the triple
  $(q,\ell,s)$ and write it to the appropriate queue $\Q_k$.
\end{itemize}
Then the segments of the current part are processed as described in
the previous section.

For each part, the algorithm reads all segments in the preceding
parts. The number of additional I/Os needed for this is
$\Oh(np/(B\log_\sigma n))$, where $p$ is the number of parts.  In
other respects, the performance of the algorithm remains essentially
the same.

We have implemented this partwise processing algorithm using greedy
on-line partitioning. That is, we make each part as large as possible
so that the peak disk usage does not exceed a given disk space
budget. An estimated peak disk usage is maintained while reading the
input.  The implementation needs at least enough disk space to store
the input (the factorization) and the output (the recovered string)
but the disk space needed in addition to that can usually be reduced
to a small fraction of the total with just a few parts.

\section{Experimental Results} \label{sec:experiments}

\begin{table*}[bt]
  \centering {
    \begin{tabular}[tab:space-basic]{l@{\hspace{2.3em}}l@{\hspace{2.3em}}r}
      \hline
      Name     & $\sigma$ & $n/z$ \\
      \hline
      hg.reads & 6     & 52.81   \\
      wiki     & 213   & 84.26   \\
      kernel   & 229   & 7767.05 \\
      random255   & 255   & 4.10    \\
      \hline
    \end{tabular} }\vspace{1ex}
  \caption[lof]{Statistics of data used in the experiments. All files are of size 256\,GiB.
    The value of $n/z$ (the average length of a phrase in the LZ77 factorization) is
    included as a measure of repetitiveness.}
  \label{tab:files}
\end{table*}

\paragraph{Setup.} We performed experiments on a machine equipped with
two six-core 1.9\,GHz Intel Xeon E5-2420 CPUs with 15\,MiB L3 cache
and 120\,GiB of DDR3 RAM.  The machine had 7.2\,TiB of disk space
striped with RAID0 across four identical local disks achieving a
(combined) transfer rate of about 480\,MiB/s.  The STXXL block size as
well as the size of buffers in the algorithm based on plain disk I/O
was set to 1\,MiB.

The OS was Linux (Ubuntu 12.04, 64bit) running kernel 3.13.0.  All
programs were compiled using {\tt g++} version 4.7.3 with {\tt-O3}
{\tt-DNDEBUG} options.  The machine had no other significant CPU tasks
running and only a single thread of execution was used for
computation.  All reported runtimes are wallclock (real) times.

\paragraph{Datasets.} For the experiments we used the following files
varying in the number of repetitions and alphabet size (see
Table~\ref{tab:files} for some statistics): \vspace{-0.2cm}
\begin{itemize}
\item hg.reads: a collection of DNA reads (short fragments produced by
  a sequencing machine) from 40 human
  genomes\footnote{\url{http://www.1000genomes.org/}} filtered from
  symbols other than $\{{\tt A}, {\tt C}, {\tt G}, {\tt T}, {\tt N}\}$
  and newline;
\item wiki: a concatenation of three different English Wikipedia
  dumps\footnote{\url{http://dumps.wikimedia.org/}} in XML format
  dated: 2014-07-07, 2014-12-08, and 2015-07-02;
\item kernel: a concatenation of $\sim$16.8 million source files from
  510 versions of Linux kernel
  \footnote{\url{http://www.kernel.org/}};
\item random255: a randomly generated sequence of bytes.
\end{itemize}

\begin{figure}[t]
  \minipage{0.53\textwidth}
  \includegraphics[trim = 0mm 15mm 0mm 0mm, clip,
  width=\linewidth]{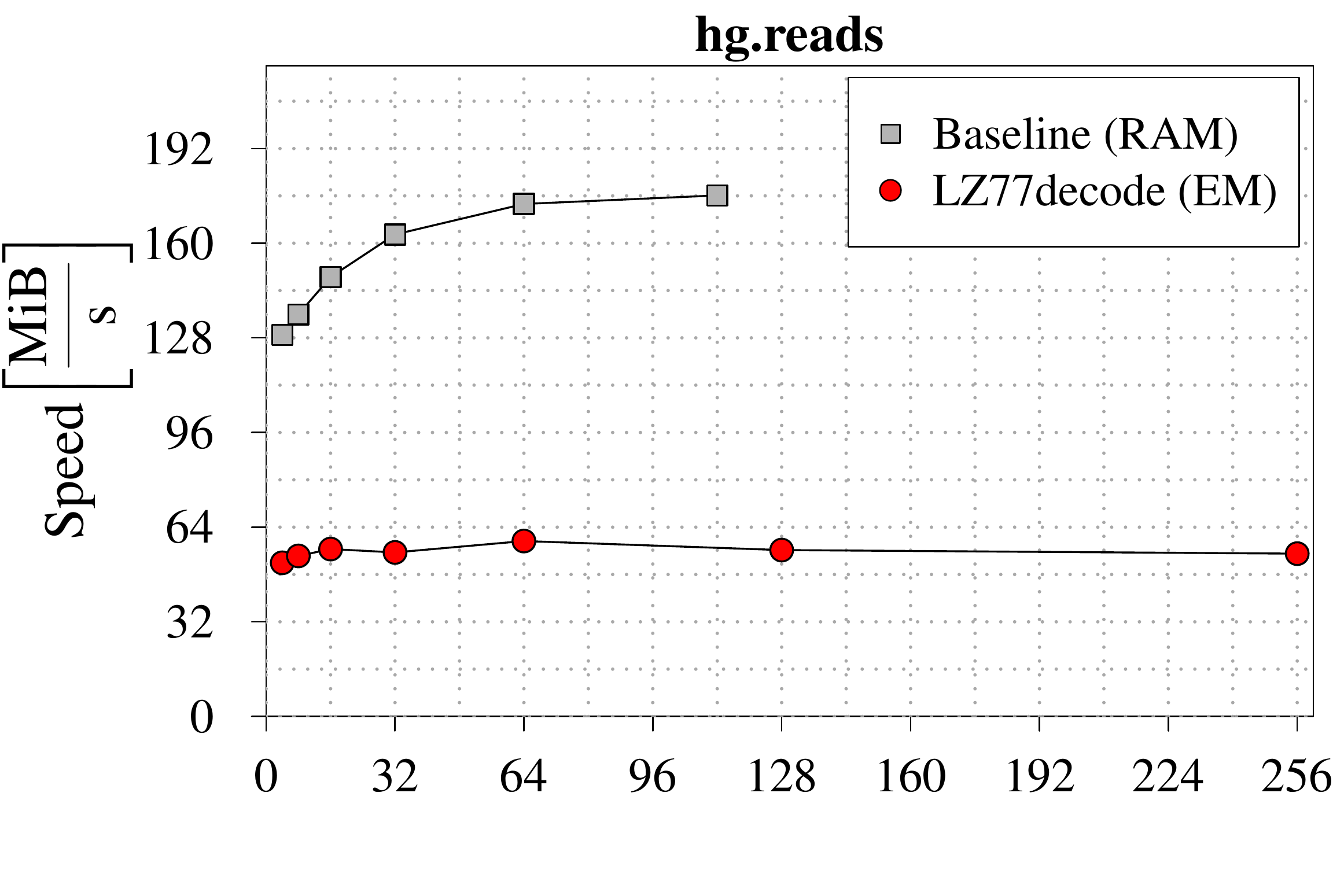}
  \endminipage
  \hspace{-0.06\textwidth} \minipage{0.53\textwidth}
  \includegraphics[trim = 0mm 15mm 0mm 0mm, clip,
  width=\linewidth]{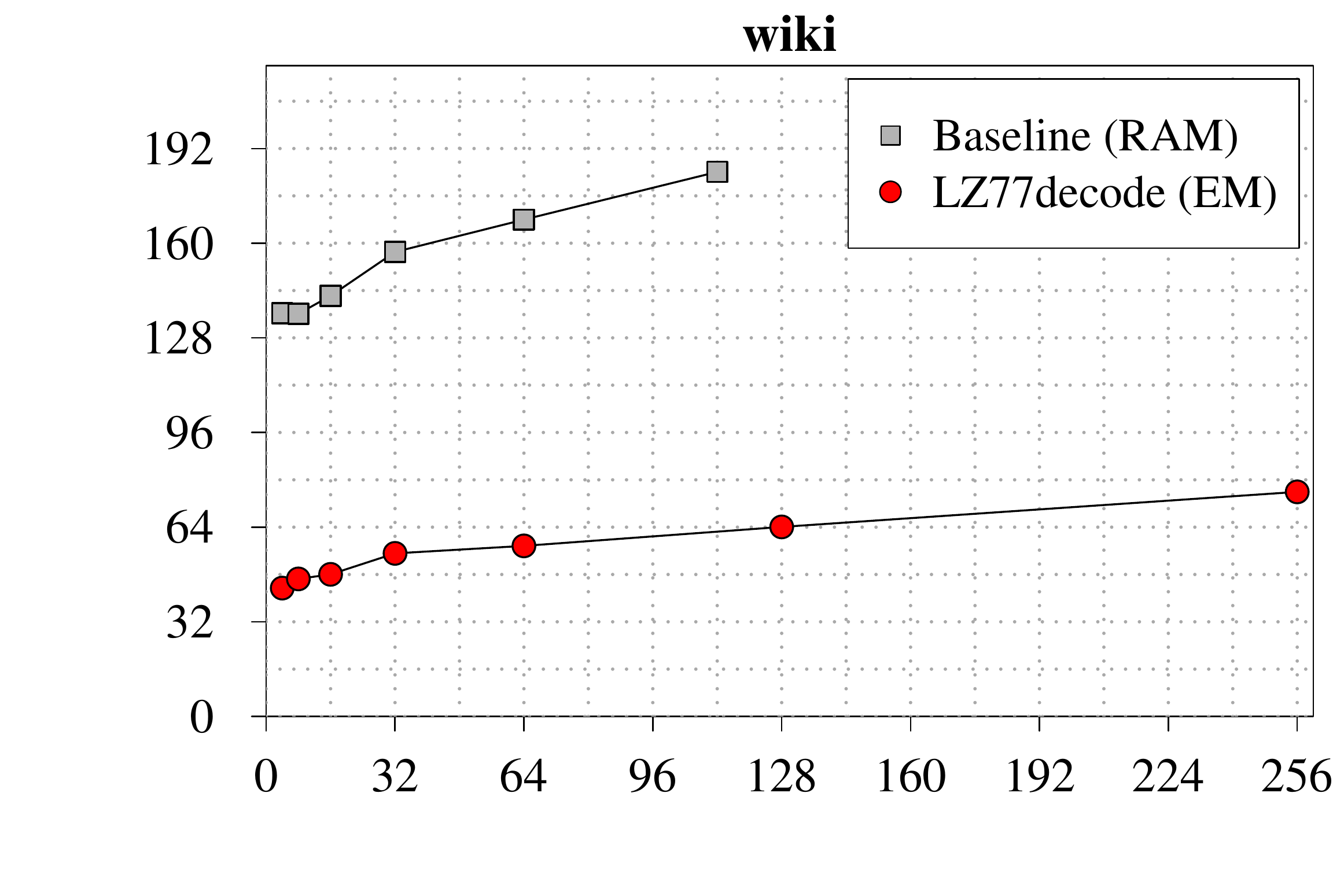}
  \endminipage
  \newline \minipage{0.53\textwidth}
  \includegraphics[trim = 0mm 0mm 0mm 0mm,
  width=\linewidth]{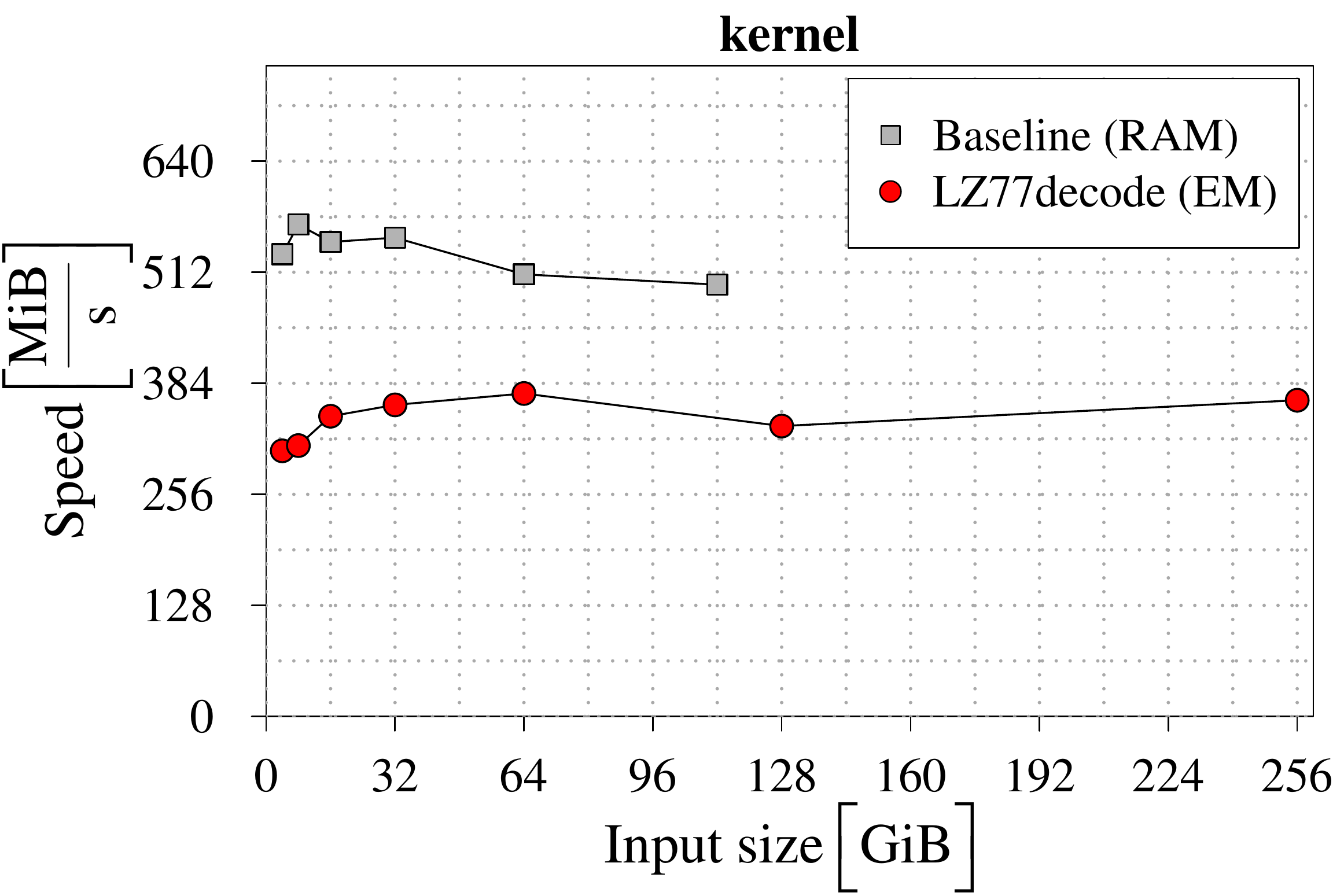}
  \endminipage
  \hspace{-0.06\textwidth} \minipage{0.53\textwidth}
  \includegraphics[trim = 0mm 0mm 0mm 0mm,
  width=\linewidth]{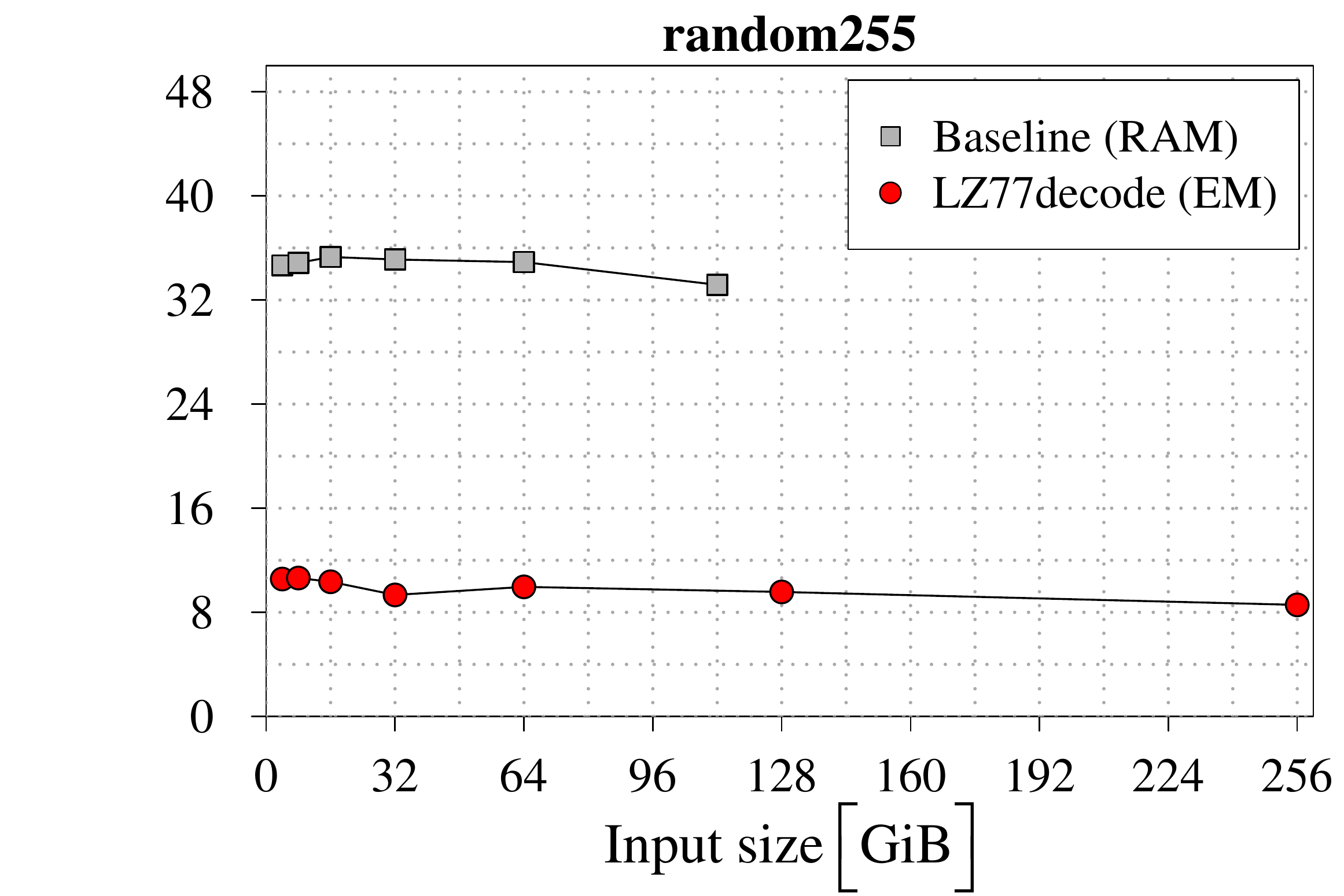}
  \endminipage
  \caption{Comparison of the new external memory LZ77 decoding
    algorithm based on plain disk I/O (``LZ77decode'') with the purely
    in-RAM decoding algorithm (``Baseline''). The latter represents an
    upper bound on the speed of LZ77 decoding. The unit of decoding
    speed is MiB of output text decoded per second.}
  \vspace{-2.5ex}
  \label{fig:scalability}
\end{figure}

\paragraph{Experiments.}
In the first experiment we compare the implementation of the new LZ77
decoding algorithm not using external-memory sorting or priority queue
to a straightforward internal-memory LZ77 decoding algorithm that
scans the input parsing from disk and decodes the text from left to
right.  All copying of text from sources to phrases happens in RAM.

We use the latter algorithm as a baseline since it represents a
realistic upper bound on the speed of LZ77 decoding. It needs enough
RAM to accommodate the output text as a whole, and thus we were only
able to process prefixes of test files up to size of about
120\,GiB. In the runtime we include the time it takes to read the
parsing from disk (we stream the parsing using a small buffer) and
write the output text to disk. The new algorithm, being fully
external-memory algorithm, can handle full test instances. The RAM
usage of the new algorithm was limited to 3.5\,GiB.

The results are presented in Fig.~\ref{fig:scalability}. In nearly all
cases the new algorithm is about three times slower than the
baseline. This is due to the fact that in the external memory
algorithm each text symbol in a far repaeat phrase is read or written
to disk three times: first, when written to a queue $\Q_j$ as a part
of a recovered phrase, second, when read from $\Q_j$, and third, when
we write the decoded text to disk. In comparison, the baseline
algorithm transfers each text symbol between RAM and disk once: when
the decoded text is written to disk.  Similarly, while the baseline
algorithm usually needs one cache miss to copy the phrase from the
source, the external memory algorithm performs about three cache
misses per phrase: when adding the source of a phrase to $\R_j$, when
adding a literal representation of a phrase into $\Q_j$, and when
copying the symbols from $\Q_j$ into their correct position in the
text.  The exception of the above behavior is the highly repetitive
kernel testfile that contains many near repeat phrases, which are
processed as efficiently as phrases in the RAM decoding algorithm.

\begin{figure}[t!]
  \minipage{0.525\textwidth}
  \includegraphics[width=\linewidth]{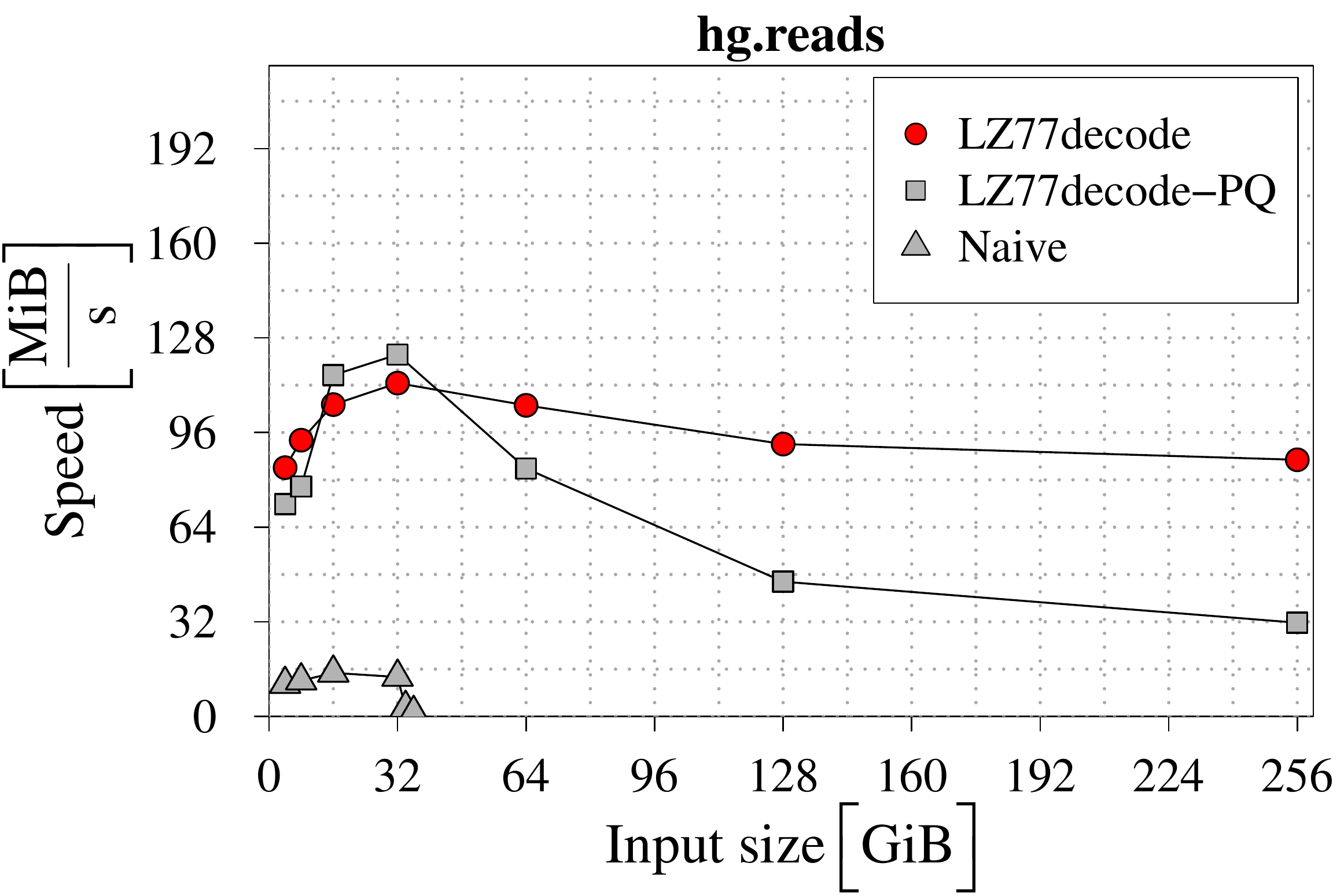}
  \endminipage
  \hspace{-0.05\textwidth} \minipage{0.525\textwidth}
  \includegraphics[width=\linewidth]{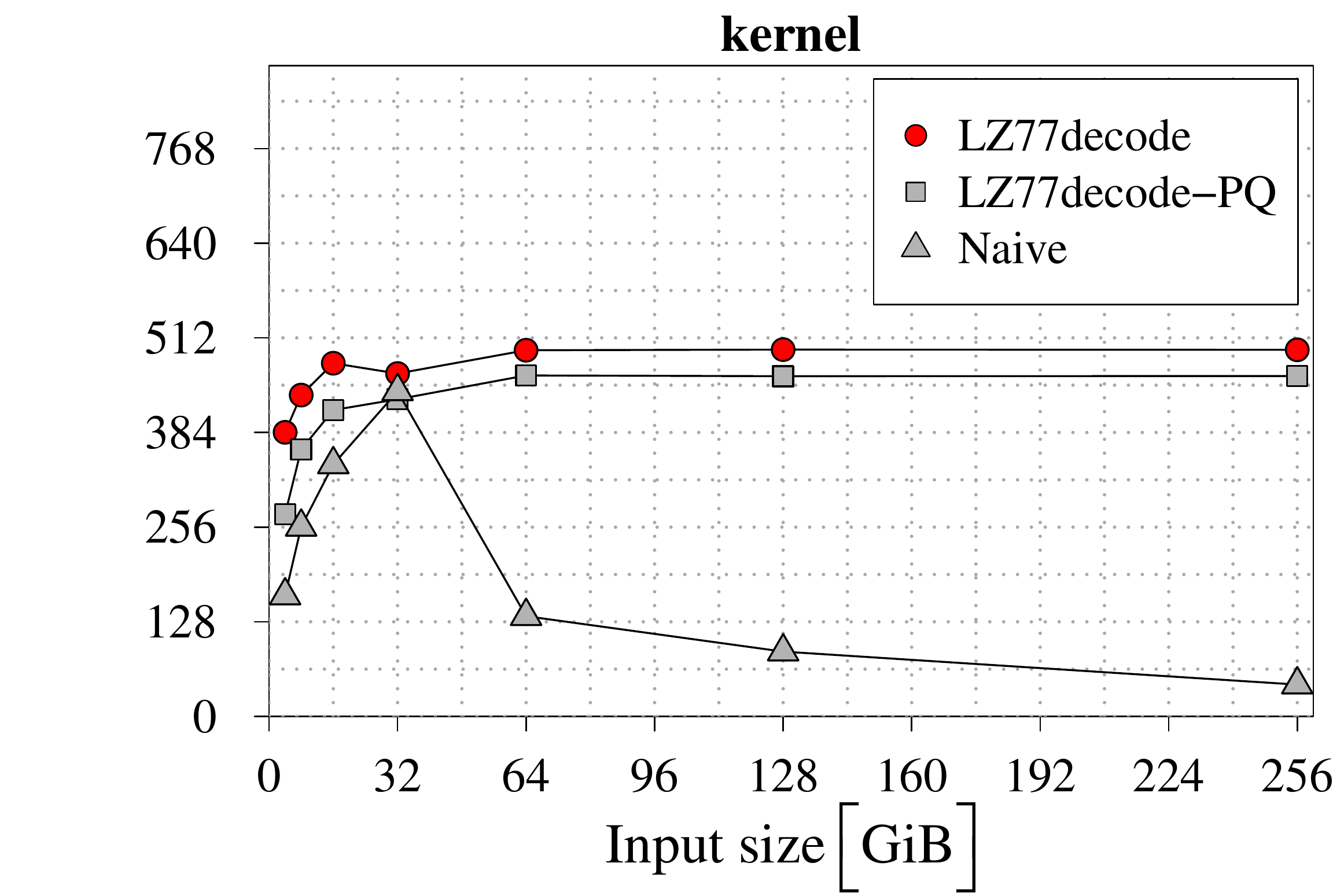}
  \endminipage
  \caption{Comparison of the new external memory LZ77 decoding
    algorithm based on plain disk I/O (``LZ77decode'') to the
    algorithm implemented using external memory sorting and priority
    queue (``LZ77decode-PQ''). The comparison also includes the
    algorithm implementing naive approach to LZ77 decoding in external
    memory. The speed is given in MiB of output text decoded per
    second.}
  \vspace{-2.5ex}
  \label{fig:versus-pq-and-naive}
\end{figure}

In the second experiment we compare our two algorithms described in
Section~\ref{sec:decode-using-sort-and-pq}
and~\ref{sec:decode-using-plain-disk-io} to each other. For the
algorithm based on priority queue we set $\ell_{\max}=16$.  The
segment size in both algorithms was set to at least half of the
available RAM (and even more if it did not lead to multiple rounds of
EM sorting/distribution), except in the algorithm based on sorting we
also need to allocate some RAM for the internal operations of STXXL
priority queue. In all instances we allocate 1\,GiB for the priority
queue (we did not observe a notable effect on performance from using
more space).

In the comparison we also include a naive external-memory decoding
algorithm that works essentially the same as baseline RAM algorithm
except we do not require that RAM is big enough to hold the
text. Whenever the algorithm requests a symbol outside a window, that
symbol is accessed from disk. We do not explicitly maintain a window
of recently decoded text in RAM, and instead achieve a very similar
effect by letting the operating system cache the recently accessed
disk pages.  To better visualize the differences in performance, all
algorithms were allowed to use 32\,GiB of RAM.

The results are given in Fig.~\ref{fig:versus-pq-and-naive}. For
highly repetitive input (kernel) there is little difference between
the new algorithms, as they both copy nearly all symbols from the
window of recently decoded text.  The naive algorithm performs much
worse, but still finishes in reasonable time due to large average
length of phrases (see Table~\ref{tab:files}).

\begin{figure}[t!]
  \minipage{0.53\textwidth}
  \includegraphics[width=\linewidth]{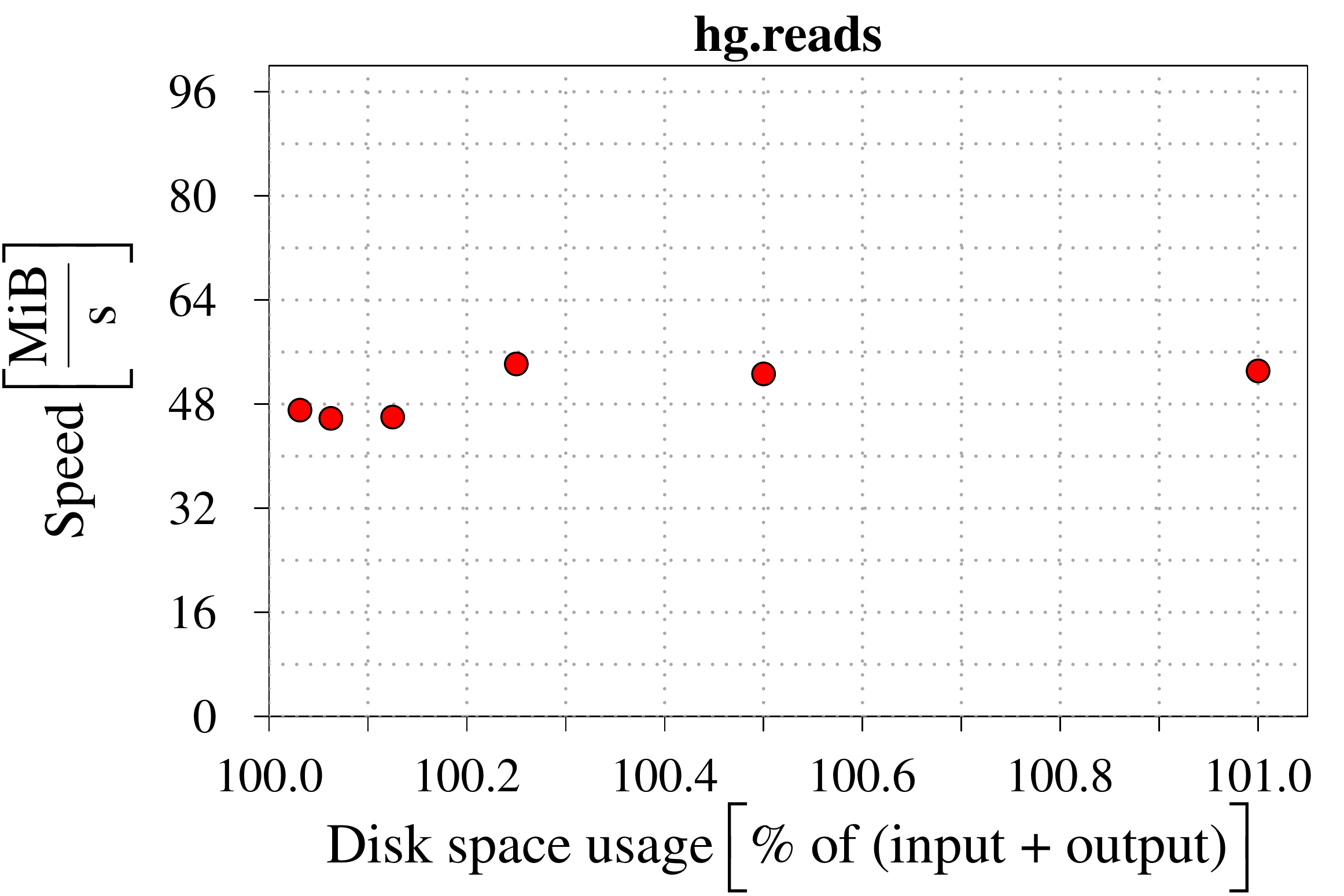}
  \endminipage
  \hspace{-0.065\textwidth} \minipage{0.53\textwidth}
  \includegraphics[width=\linewidth]{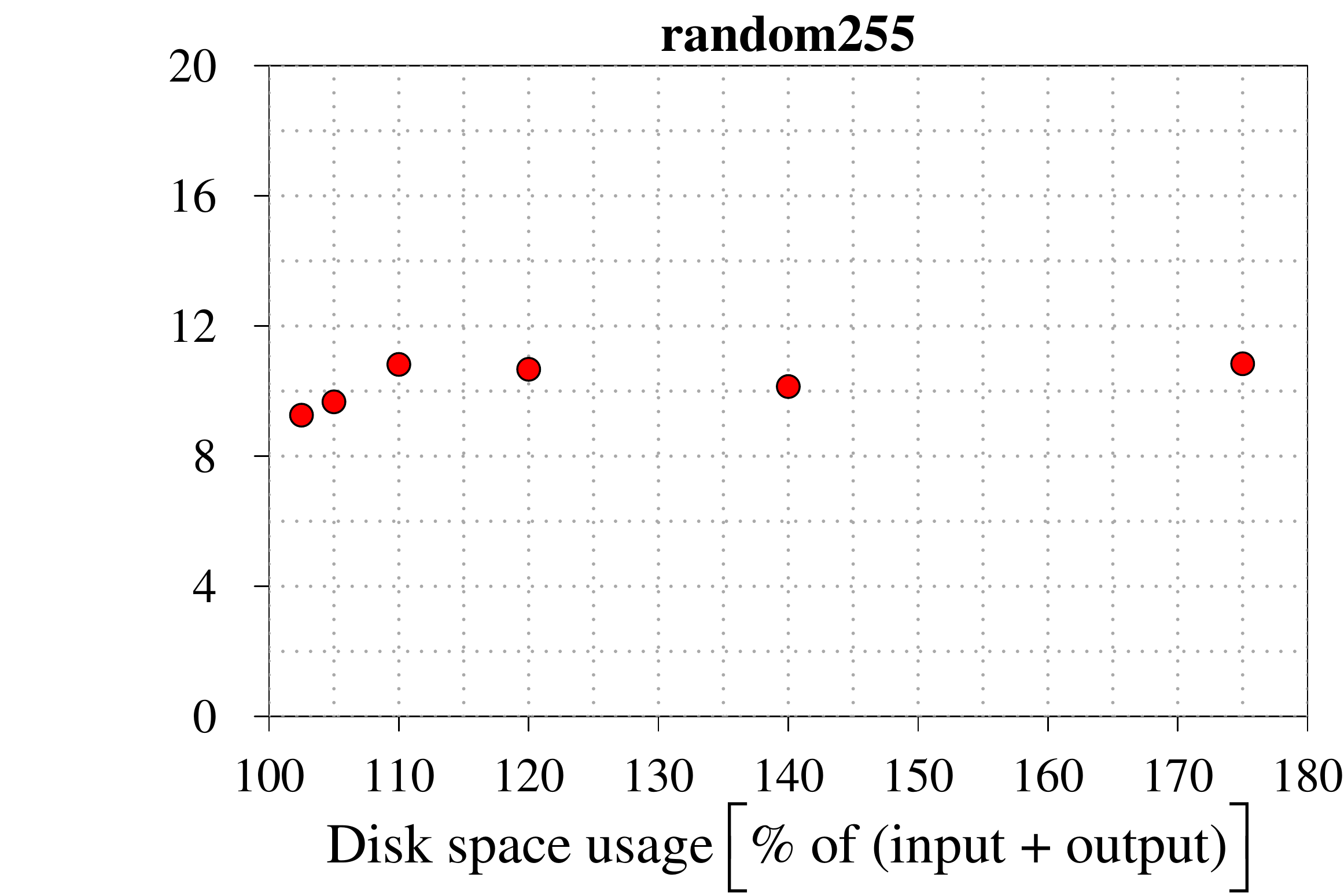}
  \endminipage
  \caption{The effect of disk space budget (see
    Section~\ref{sec:reducing-disk-space-usage}) on the speed of the
    new external-memory LZ77 decoding algorithm using plain disk I/O.
    Both testfiles were limited to 32\,GiB prefixes and the algorithm
    was allowed to use 3.5\,GiB of RAM. The rightmost data-point on
    each of the graphs represents a disk space budget sufficient to
    perform the decoding in one part.  \vspace{-2ex} }
  \label{fig:partial-processing}
\end{figure}

On the non-repetitive data (hg.reads), the algorithm using
external-memory sorting and priority queue clearly gets slower than
the algorithm using plain disk I/O as the size of input grows. The
difference in constant factors is nearly three for the largest test
instance. The naive algorithm maintains acceptable speed only up to a
point where the decoded text is larger than available RAM. At this
point random accesses to disk dramatically slow down the algorithm.

In the third experiment we explore the effect of the technique
described in Section~\ref{sec:reducing-disk-space-usage} aiming at
reducing the peak disk space usage of the new algorithm. We executed
the algorithm on 32\,GiB prefixes of two testfiles using 3.5\,GiB of
RAM and with varying disk space budgets.  As shown in
Fig.~\ref{fig:partial-processing}, this technique allows reducing the
peak disk space usage to very little over what is necessary to store
the input parsing and output text and does not have a significant
effect on the runtime of the algorithm, even on the incompressible
random data.

\section{Concluding Remarks}

We have described the first algorithms for external memory LZ77
decoding. Our experimental results show that LZ77 decoding is fast in
external memory setting too. The state-of-the-art external memory LZ
factorization algorithms are more than a magnitude slower than our
fastest decoding algorithm, see~\cite{KKP2014-EMLZ}.

\bibliographystyle{splncs03}
\bibliography{lzdec}

\end{document}